\documentclass{new_tlp}

\usepackage{amsmath}
\usepackage{amssymb}
\usepackage{bussproofs}
\usepackage{multicol}
\usepackage{hyperref}
\usepackage{microtype}
\usepackage{thmtools}
\usepackage{mathtools}
\newcommand{\coloncolonequals}{\Coloneqq}

\makeatletter
\@ifundefined{lhs2tex.lhs2tex.sty.read}%
  {\@namedef{lhs2tex.lhs2tex.sty.read}{}%
   \newcommand\SkipToFmtEnd{}%
   \newcommand\EndFmtInput{}%
   \long\def\SkipToFmtEnd#1\EndFmtInput{}%
  }\SkipToFmtEnd

\newcommand\ReadOnlyOnce[1]{\@ifundefined{#1}{\@namedef{#1}{}}\SkipToFmtEnd}
\usepackage{amstext}
\usepackage{amssymb}
\usepackage{stmaryrd}
\DeclareFontFamily{OT1}{cmtex}{}
\DeclareFontShape{OT1}{cmtex}{m}{n}
  {<5><6><7><8>cmtex8
   <9>cmtex9
   <10><10.95><12><14.4><17.28><20.74><24.88>cmtex10}{}
\DeclareFontShape{OT1}{cmtex}{m}{it}
  {<-> ssub * cmtt/m/it}{}

\DeclareFontShape{OT1}{cmtt}{bx}{n}
  {<5><6><7><8>cmtt8
   <9>cmbtt9
   <10><10.95><12><14.4><17.28><20.74><24.88>cmbtt10}{}
\DeclareFontShape{OT1}{cmtex}{bx}{n}
  {<-> ssub * cmtt/bx/n}{}

\newlength{\lwidth}
\newlength{\cwidth}

\newcommand{\Conid}[1]{\mathit{#1}}
\newcommand{\Varid}[1]{\mathit{#1}}
\newcommand{\anonymous}{\kern0.06em \vbox{\hrule\@width.5em}}

\renewcommand{\leq}{\leqslant}

\EndFmtInput
\makeatother

\newcommand{\chrnabla}{CHR${}^\nabla$}

\newcommand{\hide}[1]{}

\newtheorem{defn}{Definition}

\newtheorem{lemma}{Lemma}
\newtheorem{corollary}{Corollary}

\title{Constraint Handling Rules with Binders, Patterns and Generic Quantification}

\author[A. Serrano and J. Hage]{ALEJANDRO SERRANO and JURRIAAN HAGE \thanks{This work was supported by the Netherlands Organisation for Scientific Research (NWO) project on ``DOMain Specific Type Error Diagnosis (DOMSTED)'' (612.001.213).} \\
 Department of Information and Computing Sciences, Utrecht University \\
 \email{\{A.SerranoMena, J.Hage\}@uu.nl}}

\jdate{August 2017}
\pubyear{2017}
\pagerange{\pageref{firstpage}--\pageref{lastpage}}
\doi{-}

\begin{document}

\label{firstpage}

\maketitle

\begin{abstract}
Constraint Handling Rules provide descriptions for constraint solvers. However, they fall short when those constraints specify some binding structure, like higher-rank types in a constraint-based type inference algorithm. In this paper, the term syntax of constraints is replaced by $\lambda$-tree syntax, in which binding is explicit; and a new $\nabla$ generic quantifier is introduced, which is used to create new fresh constants.

\noindent This paper is under consideration for publication in TPLP.
\end{abstract}

\section{Introduction}
\label{sec:intro}

Constraint Handling Rules \cite{Frhwirth:2009:CHR:1618539} -- usually shortened to CHRs -- provide a language to describe constraint solvers. The great body of work related to CHRs, summarized in \cite{Fruhwirth199895,DBLP:journals/tplp/SneyersWSK10}, makes them a good choice to describe algorithms where constraint rewriting is involved.

Type inference for functional languages is one of the areas in which CHRs have been applied successfully. CHRs have been used to improve type error reporting \cite{ChameleonProcessing,wazny_phdthesis06,DBLP:conf/esop/SerranoH16}, describe and extend the type class machinery in Haskell \cite{Sulzmann:2007:UFD:1194875.1194877,localinstances} and generalize the shape of algebraic data types \cite{Sulzmann:2006:FEA:2100071.2100078}. However, CHRs are not enough to describe in a concise form\footnote{Since CHRs are Turing-complete \cite{Sneyers:2009:CPC:1462166.1462169}, in theory any algorithm can be described. This does not mean that the description is convenient, though.} the solving needed for some advanced type system features.

In particular, consider parametric polymorphism. In languages based on the Hindley-Damas-Milner typing discipline, such as Haskell and ML, we find a stratification of types. Expressions are assigned a simple type, whereas declarations are given a \emph{type scheme} which quantifies over a set of variables. For example, the identity function \ensuremath{\Varid{id}} types with the following type scheme: \ensuremath{\forall\!\;\Varid{a}.\,\Varid{a}\to \Varid{a}}.

At each point where the \ensuremath{\Varid{id}} function is used as an expression, we cannot directly use this type scheme, since expressions must be assigned a type. Thus, we are forced to remove the quantification by \emph{instantiating} \ensuremath{\Varid{a}}. In other words, every occurrence of \ensuremath{\Varid{id}} is assigned a type \ensuremath{\alpha\to \alpha} for a fresh variable \ensuremath{\alpha}. Solving is responsible for finding a type for that variable, or deciding to quantify over it.

One of the beauties of Hindley-Damas-Milner is that we can instantiate the type schemes in a binding block at once before we start solving. Instantiating variables during constraint gathering takes us quite far: the type system of Haskell as of GHC 7, including type classes, type families and GADTs has been described in this fashion \cite{outsidein}. However, when type application \cite{Eisenberg:2016:VTA:2958880.2958890} or higher-rank types \cite{DBLP:journals/jfp/JonesVWS07} enter the picture we need to \emph{delay} instantiation.

As a result, whereas before each use of \ensuremath{\Varid{id}} in the source code would lead to the type \ensuremath{\alpha\to \alpha} for a fresh variable \ensuremath{\alpha}, now we simply assign it a fresh type \ensuremath{\beta} and recall the instantiation relation by means of a constraint \ensuremath{\forall\!\;\Varid{a}.\,\Varid{a}\to \Varid{a}\leq \beta}. The ball of finding a good type assignment for \ensuremath{\beta} is now in the court of the constraint solver, which we describe using CHRs. A rule instantiating polymorphic types looks similar to:
\begin{tabbing}
\qquad\=\hspace{\lwidth}\=\hspace{\cwidth}\=\+\kill
${\mathit{forall}\;(\Conid{X},\Conid{T})\leq \Conid{S}\iff\Conid{T'}\mathrel{=}\Varid{inst}\;(\Conid{X},\Conid{T})\mid \Conid{T'}\mathrel{=}\Conid{S}}$
\end{tabbing}intutively, this rule states that every time it finds an \ensuremath{\Conid{R}\leq \Conid{S}} constraint where the left-hand side is polymorphic -- \ensuremath{\mathit{forall}\;(\Conid{X},\Conid{T})} -- it should replace it by a type equality \ensuremath{\Conid{T'}\mathrel{=}\Conid{S}} where \ensuremath{\Conid{T'}} is derived from \ensuremath{\Conid{R}} by instantiating \ensuremath{\Conid{X}}. There are some problems with this rule:
\begin{enumerate}
\item There is no guarantee that variables are correctly bound in polymorphic types. We need to rely on the invariant that we are allowed to use variable \ensuremath{\Conid{X}} only inside \ensuremath{\Conid{T}}.
\item How to adequately represent type variables within CHRs is not yet well understood \cite{Csorba_prosand}. The simplest solution, using a unique number per variable, does not work well with binding structures.
\item The \ensuremath{\Varid{inst}} function has to be defined \emph{externally} to the constraint solver. Generation of fresh variables needs some extra state not reflected directly in the rules. Furthermore, instantiation is tricky due to variable capture.
\end{enumerate}
In this paper we propose an extension of CHR, which we call \chrnabla, in which the binding structure of a term is explicit and represented by $\lambda$-abstractions: \ensuremath{\mathit{forall}\;(\lambda \Conid{X}.\,\Conid{T})}. This syntax for binders is called $\lambda$-tree syntax \cite{Miller2000}. The previous constraint transformation is now written as:
\begin{tabbing}
\qquad\=\hspace{\lwidth}\=\hspace{\cwidth}\=\+\kill
${(\mathrm{1})\;\qquad\;\mathit{forall}\;(\Conid{T})\leq \Conid{S}\iff\exists\!\;\Conid{V}.\,\Conid{T}\;\Conid{V}\leq \Conid{S}}$
\end{tabbing}In the body we create a new variable \ensuremath{\Conid{V}} by means of the \ensuremath{\exists\!} \,\! operator, and then we replace \ensuremath{\Conid{X}} with \ensuremath{\Conid{V}} in the body of \ensuremath{\Conid{T}}. The system ensures that no variable is incorrectly captured.

We can go a step further and consider the shape of the rule whenever a polymorphic type is present to the right of the \ensuremath{\leq } constraint. The usual definition of instantiation says that $\sigma \leq \forall a. \tau$ if and only if $\sigma \leq \tau[a \mapsto \rho]$ for any choice of $\rho$. One way to prove this fact is by introducing a \emph{nominal constant} $\hat{a}$ standing for $a$, that is, a constant distinct from any other term in the language, and proving $\sigma \leq \tau[a \mapsto \hat{a}]$. Nominal constraints are also referred to as rigid or Skolem variables (we use the three names interchangeably throughout the paper). Our extension to CHRs includes a $\nabla$ operator, inspired by the logic of Abella \cite{TIU20073,JFR4650}, which introduces a new nominal constant.
\begin{tabbing}
\qquad\=\hspace{\lwidth}\=\hspace{\cwidth}\=\+\kill
${(\mathrm{2})\;\qquad\;\Conid{T}\leq \mathit{forall}\;(\Conid{S})\iff\nabla\!\;\Conid{A}.\,\Conid{T}\leq \Conid{S}\;\Conid{A}}$
\end{tabbing}Note that in order to have a sound algorithm, rule (1) must always be applied before rule (2). User-definable rule priorities \cite{DeKoninck:2007:URP:1273920.1273924} add support for preferences in CHRs. Priorities, on the other hand, are orthogonal to our extensions to CHR, and thus we concern ourselves only with ``classical'' rules.

\paragraph{Contributions.} Specifically, our contributions in this paper are:
\begin{itemize}
\itemsep0em
\item Extending CHR matching from ground terms to $\lambda$-tree terms via $L_\lambda$-unification.
\item Incorporating the notion of nominal constants and the $\nabla$ operator to generate fresh instances of these constants.
\item Providing techniques to deal with confluence and termination in this new scenario.
\end{itemize}
As an example of the power of the framework, we showcase an extension to the Haskell language to provide simple higher-rank types.

The integration of $\lambda$-tree syntax with $\nabla$ was already present in the Abella theorem prover \cite{JFR4650}, the integration with CHRs is entirely novel.


\section{Preliminaries}

In this section we give a brief introduction to each framework involved in our work.

\subsection{Constraint Handling Rules}
\label{sec:chr}

The language of CHRs has three kinds of rules:
$$\begin{array}{cccccr}
    & & H^r & \; \iff \;     & G \;\; | \;\; B & \quad \textrm{simplification}\\
H^k & &     & \; \implies \; & G \;\; | \;\; B & \quad \textrm{propagation} \\
H^k & \backslash & H^r & \; \iff \;  & G \;\; | \;\;  B & \quad \textrm{simpagation} \\
\end{array}$$
In each case, $H^k$, $H^r$ and $B$ are sets of constraints, called the heads and the body respectively. We use $\top$ to represent an empty set of constraints (reminescent of ``true''). In order for a rule to be applied, some constraints from the current set must match the heads, and the guard $G$ must be satisfied. Rewriting depends on the kind of rule: with simplification rules the constraints $H^r$ are replaced by $B$, in propagation rules the constraints $B$ are added to the set but the constraints $H^k$ are kept. Simpagation rules are a generalization of both: $H^k$ constraints are kept and $H^r$ are removed. In fact, we can view any CHR as a simpagation rule where the heads might be empty.

CHRs are applied non-deterministically. For a given initial constraint set, many different sequences of applications of rules are usually possible. Confluence, that is, the fact that the outcome of the process does not depend on the order in which rules are applied, must be proven externally by the author of the CHRs.

Type classes in Haskell are a prime example of what can be described using the CHR language \cite{Sulzmann:2007:UFD:1194875.1194877}. A declaration of the following type classes and instances:
\begin{center}
\begin{tabular}{ll}
\ensuremath{\mathbf{class}\;\Conid{Eq}\;\Varid{a}\;\mathbf{where}\mathbin{...}} & \ensuremath{\mathbf{class}\;\Conid{Eq}\;\Varid{a}\Rightarrow \Conid{Ord}\;\Varid{a}\;\mathbf{where}\mathbin{...}} \\
\ensuremath{\mathbf{instance}\;\Conid{Eq}\;\Conid{Int}\;\mathbf{where}\mathbin{...}} & \ensuremath{\mathbf{instance}\;\Conid{Ord}\;\Conid{Int}\;\mathbf{where}\mathbin{...}} \\
\ensuremath{\mathbf{instance}\;\Conid{Eq}\;\Varid{a}\Rightarrow \Conid{Eq}\;\,\;\!\![\Varid{a}\mskip1.5mu]\;\mathbf{where}\mathinner{\ldotp\ldotp}} & \ensuremath{\mathbf{instance}\;\Conid{Ord}\;\Varid{a}\Rightarrow \Conid{Ord}\;\,\;\!\![\Varid{a}\mskip1.5mu]\;\mathbf{where}\mathbin{...}}
\end{tabular}
\end{center}
gives rise to the following set of rules:
\begin{center}
\begin{tabular}{ll}
& \ensuremath{\Conid{Ord}\;\Conid{A}\implies\Conid{Eq}\;\Conid{A}} \\
\ensuremath{\Conid{Eq}\;\Conid{Int}\iff\top} & \ensuremath{\Conid{Ord}\;\Conid{Int}\iff\top}  \\
\ensuremath{\Conid{Eq}\;\,\;\!\![\Conid{A}\mskip1.5mu]\iff\Conid{Eq}\;\Conid{A}} & \ensuremath{\Conid{Ord}\;\,\;\!\![\Conid{A}\mskip1.5mu]\iff\Conid{Ord}\;\Conid{A}} \\
\end{tabular}
\end{center}
Using these rules we have different ways to get from \ensuremath{\Conid{Ord}\;\,\;\!\![\Varid{a}\mskip1.5mu]} to \ensuremath{\Conid{Eq}\;\Varid{a}}, illustrating the non-deterministic nature of CHRs. On the one hand, \ensuremath{\Conid{Ord}\;\,\;\!\![\Varid{a}\mskip1.5mu]} can be simplified to \ensuremath{\Conid{Ord}\;\Varid{a}}, which then generates a constraint \ensuremath{\Conid{Eq}\;\Varid{a}}. On the other hand, the generation of \ensuremath{\Conid{Eq}\;\,\;\!\![\Varid{a}\mskip1.5mu]} might take place first, and only then will it be simplified to \ensuremath{\Conid{Eq}\;\Varid{a}}. The restrictions that the Haskell language imposes on type classes ensure that both paths are equivalent. When translated to CHRs, the resulting rules form a confluent set. 

\subsection{$\lambda$-Tree Syntax, $L_\lambda$-Unification and $\beta_0$-Reduction}
\label{sec:lambdatree}

Several approaches exist for the representation of binding inside a language in a both convenient and efficient way, including de Bruijn indices \cite{DEBRUIJN1972381}, locally nameless representation \cite{Chargueraud2012}, and extensions dealing with hygiene in macros and different namespaces. In this paper we use $\lambda$-tree syntax \cite{Miller2000}, which is closely related to higher-order abstract syntax \cite{Pfenning:1988:HAS:960116.54010}.

Consider $\lambda$-terms, which extend a base language of terms with a binding operator $\lambda$, variables $x$ and an application form $T_1 \; T_2$. The equality relation between $\lambda$-terms includes not only syntactic equality, but also the following three rules:
$$\begin{array}{rclrr}
\lambda x. \, B & = & \lambda y. \, B[x \mapsto y] & \text{if $y$ not free in $B$} & \quad (\alpha) \\
(\lambda x. \, B) \, E & = & \multicolumn{2}{l}{B[x \mapsto E] \; \text{if $E$ does not contain $x$}} & (\beta) \\
\lambda x. \, F \; x & = & F & & (\eta) \\ 
\end{array}$$

During solving, the CHR engine needs to check whether any subset of the active constraints matches a rule. In order to do so, it matches the constraints with the patterns appearing in the head of the rule. In the setting of normal CHRs using first-order terms this check is unification. But now we have $\lambda$-terms, leading to higher-order unification \cite{HUET197527}. Full higher-order unification has several drawbacks, including undecidability. Miller \cite{Miller2000} argues for a weaker matching procedure, but which guarantees decidability, finds most general unifiers and runs in linear time \cite{Qian1993}. This procedure is called $L_\lambda$-unification, or (higher-order) pattern unification.

The main restriction in pattern unification is that an application in which the head is a metavariable must be done to distinct bound variables. Using a variable repeatedly is not allowed, so matching with \ensuremath{\lambda \Varid{x}.\,\lambda \Varid{y}.\,\Conid{F}\;\Varid{x}\;\Varid{y}} is OK, but \ensuremath{\lambda \Varid{x}.\,\Conid{F}\;\Varid{x}\;\Varid{x}} is not. Using a term which is not a variable as an argument, as in \ensuremath{\lambda \Varid{x}.\,\Conid{F}\;(\Conid{G}\;\Varid{x})} is not allowed either. The only exception is the $\eta$-expansion of variables: the term \ensuremath{\lambda \Varid{x}.\,\Conid{F}\;(\lambda \Varid{z}.\,\Varid{x}\;\Varid{x})} is allowed, since it is equivalent to \ensuremath{\lambda \Varid{x}.\,\Conid{F}\;\Varid{x}}.

The theory of unification of patterns is actually a theory of equality of $\lambda$-terms with the above restrictions. In such a scenario, the full power of $\beta$-reduction is not needed, but just a restricted version for variables.
$$\begin{array}{rclr}
(\lambda x. \, B) \, y & = & B[x \mapsto y] & \quad (\textrm{rule } \beta_0) \\
\end{array}$$

The reader may be worried about patterns being overly restrictive for our purposes. \cite{Miller91alogic} argues that for practical purposes pattern unification is enough, citing developments in the Isabelle theorem prover and the $\lambda$Prolog logic system. Later developments, such as the Abella theorem prover, also make use of $L_\lambda$-unification. Thus, patterns seem to be a sweet spot to base a language on.

\subsection{The Generic Quantifier $\nabla$}
\label{sec:nablaintro}

As Miller and Tiu discuss in \cite{Miller:2005:PTG:1094622.1094628}, there are two ways to prove a universally quantified proposition $\forall x. F$. The first one is proving $F[x \mapsto T]$ for every closed term $T$. Usually the set of terms is inductively defined, in which case \emph{induction} can be used to reduce the number of cases. Another possibility is proving $F[x \mapsto c]$ for a completely new nominal constant $c$ which shall remain fresh during the whole proof.

These two notions are not completely interchangeable: $\forall x \, y. P(x,y) \implies \forall z. P(z,z)$ holds for the first approach, since you can instantiate $x$ and $y$ in the first universal with the same value $z$. But $\forall x \, y. P(x,y) \implies \forall z. P(z,z)$ does \emph{not} hold in general with the second reading: we instantiate the antecedent with two different constants, leading to $P(a, b)$. But now we cannot make $a$ equal to $b$ -- which we need to prove $\forall z. P(z,z)$ -- all we know about those constants is indeed that they are different!

In order to distinguish these different mechanics, \cite{Miller:2005:PTG:1094622.1094628} introduces a new \emph{generic quantifier $\nabla$} to account for the second nature. During proof search, $\nabla$ introduces a new scoped constant, different from any other such constant in the proof. The resulting logic was extended \cite{TIU20073} to account for some desirable properties of $\nabla$, such as $\nabla x. \, B$ being equivalent to $B$ whenever $x$ is not free in $B$. The resulting logic, $LG^\omega$, forms the basis of our work. Formally, the rules governing this new universal quantifier $\nabla$ are:
\begin{prooftree}
\AxiomC{$\Gamma, B[x \mapsto \mathtt{a}] \vdash C$}
\AxiomC{$a \not\in \mathsf{supp}(B)$}
\RightLabel{\sc $\nabla$l}
\BinaryInfC{$\Gamma, \nabla x. B \vdash C$} 
\end{prooftree}
\begin{prooftree}
\AxiomC{$\Gamma \vdash C[x \mapsto a]$}
\AxiomC{$a \not\in \mathsf{supp}(C)$}
\RightLabel{\sc $\nabla$r}
\BinaryInfC{$\Gamma \vdash \nabla x. C$} 
\end{prooftree}
\begin{prooftree}
\AxiomC{$\pi(B) \equiv \pi'(B')$}
\AxiomC{$\pi, \pi'$ permutations of constants}
\RightLabel{\sc id$\pi$}
\BinaryInfC{$\Gamma, B \vdash B'$} 
\end{prooftree}
In the rules, the \emph{support} of a formula $B$, $\mathsf{supp}(B)$, is defined as the set of scoped constants in $B$. The reader can see that left and right-introduction rules for $\nabla$ work in the same way: $\nabla$ is indeed a self-dual quantifier.

Having only these two rules is not enough to prove some of the theorems we would like to hold; in particular $\nabla x. B(x) \vdash \nabla y. B(y)$ is not true. The reason is that each $\nabla$ introduces its own fresh constant, which are guaranteed to be distinct. More of these ``non-examples'' can be found in Figure 4 of \cite{Miller:2005:PTG:1094622.1094628}. The solution is to allow a permutation of constants in both the antecedent and the consequent, as shown in rule {\sc id$\pi$}. In our case, assume that after the introduction of constants we have to prove $B(\mathtt{a}) \vdash B(\mathtt{b})$ for constants $\mathtt{a}$ and $\mathtt{b}$. By applying an identity permutation on the left and the permutation $[\mathtt{a} \mapsto \mathtt{b}, \mathtt{b} \mapsto \mathtt{a}]$ on the right, we obtain syntactically equal formulas.

The previous rules make $\nabla$ commute with the $\vee$, $\wedge$ and $\supset$ connectives. Swapping of constants respects provability: $\nabla x. \nabla y. B(x,y) \equiv \nabla y. \nabla x. B(x,y)$, as witnessed by one application of {\sc id$\pi$}. Finally, we have that $\forall x. B(x) \vdash \nabla x. B(x)$ and $\nabla x. B(x) \vdash \exists x. B(x)$.

The problem which led to the creation of $\nabla$ was to reason about an object logic in a different meta-logic \cite{TIU20073}. This idea fits our problem: we want to reason about constraints with universal quantification, our object logic, at a higher level, namely CHRs. As we shall see in \autoref{sec:nabla}, each appearance of a $\nabla$ in a rule leads to the introduction of a new scoped constant.

\section{Putting All Together: \chrnabla}
\label{sec:nabla}

In this section we introduce our extensions to the CHR machinery needed to cope with rules such as those in the introduction. A simple implementation of type inference for higher-rank types is presented as an example of its use.

\paragraph{Syntax.}

\begin{figure}[t]
\begin{center}
\begin{tabular}{lrclr}
Variables   & $\mathfrak{V}$ & $\ni$ & $X, Y, Q, R, T, V \dots$ \\
\vspace{0.1cm} Nominal constants & $\mathfrak{K}$ & $\ni$ & $\mathtt{a}, \mathtt{b}, \dots$ \\
Constraints / patterns & $\quad C$ & $\coloncolonequals$ & $c(T_1, \dots, T_n)$ & $\quad c \in \mathfrak{C}$ \\
Terms       & $\quad T$ & $\coloncolonequals$ & $X \;\; | \;\; \mathtt{a}$ \\
            &           & $|$                & $f(T_1, \dots, T_n)$ & $\quad f \in \mathfrak{F}$ \\
            &           & $|$                & $\lambda X. T$ \\
            &           & $|$                & $T_1 \; T_2$ & \\
\end{tabular}
\end{center}
\caption{Syntax of constraints and terms}
\label{fig:syntax}
\end{figure}

In the syntax of \chrnabla, we use four sets of objects. These sets must be disjoint, except for constraint constructors and term constructors, which may overlap.
\begin{itemize}
\itemsep0em
\item A set of \emph{constraint} constructors $\mathfrak{C}$ annotated with arity.
\item A set of \emph{term} constructors $\mathfrak{F}$ annotated with their arity. We denote both types of constructors by lowercase letters such as $c, f, g, \dots$
\item An infinite set of \emph{term variables} $\mathfrak{V}$, which we denote by uppercase letters $X, Y, \dots$
\item An infinite set of \emph{nominal constants} $\mathfrak{K}$, which we denote by teletype letters $\mathtt{a}, \mathtt{b}, \dots$
\end{itemize}
Using these sets, we build up both constraints and terms, as given in \autoref{fig:syntax}. Note that at the constraint level, abstraction and application are not permitted; this richer syntactic structure is only available to terms.

In \chrnabla{} the syntax of terms and patterns coincide; we use the second term to emphasize the role of a term being part of a rule. In the following we often use $C$ to refer to both single constraints or sets of them; the context is enough to distinguish the intented meaning.

One of our goals in this paper is to allow the more powerful $L_\lambda$-unification to be used instead of plain term unification. In order to do so, we need to restrict the shape of constraints which may appear, as described in \autoref{sec:lambdatree}.
\begin{defn}[Well-defined patterns and constraints]
We say that a pattern $P$ is \emph{well-defined} if and only if free variables appear only applied to distinct variables or $\eta$-equivalent versions of these.

\noindent We say that a constraint $C$ is \emph{well-defined} if and only if each of its terms is a well-defined pattern and does not contain any nominal constant.
\end{defn}

Finally we can describe the syntax of a rule in \chrnabla:
$$H^k \; \backslash \; H^r \iff  G \;\; | \;\; \nabla X_1 \dots X_n. \, \exists Y_1 \dots Y_m. \, B$$
where $H^k$, $H^r$ are sets of well-defined constraints, and all free variables in $B$ come from either the free variables in $H^k$ and $H^r$ or from $\{ X_1, \dots, X_n, Y_1, \dots, Y_m \}$. Both the set of universally quantified and of existentially quantified variables may be empty. Note that all variables in a set of constraints must come from either the initial set or introduced by explicit quantification. In contrast, ``classical'' CHRs quantify variables implicitly.

We define simplification and propagation rules as a restriction of the main kind of rule with empty $H^k$ and $H^r$, respectively, as shown in \autoref{sec:chr}.

\paragraph{Declarative semantics.} There are different ways to interpret a set of CHRs. In other words, we can attach different \emph{semantics} to them. The \emph{declarative} semantics \cite{Fruhwirth199895} maps each rule to a logic formula. First, let us consider the declarative semantics of a rule without any of our extensions, that is,
$H^k \; \backslash \; H^r \iff  G \;\; | \;\; \exists \bar{Y}. B$.
Let $\bar{Z}$ be the set of free variables in $H^k$ and $H^r$; all the variables in $B$ are elements of $\bar{Z} \cup \bar{Y}$. The declarative semantics of such rule is defined as:
$$\forall \bar{Z}. (H^k \wedge G) \supset (H^r \leftrightarrow \exists \bar{Y}. B)$$
The declarative semantics of a rule $H^k \; \backslash \; H^r \iff  G \;\; | \;\; \nabla \bar{X}. \exists \bar{Y}. B$ looks similar,
$$\forall \bar{Z}. (H^k \wedge G) \supset (H^r \leftrightarrow \nabla \bar{X}. \exists \bar{Y}. B)$$
However, notice that in a ``classical'' rule the quantified variables range over terms, whereas in \chrnabla{} they range over $\lambda$-trees. This means that abstraction and application are also allowed by the syntax, and that $\alpha$, $\beta_0$ and $\eta$ rules relate equivalent constraints.

\paragraph{Theoretical operational semantics.} The other common semantics for CHRs is the so-called \emph{theoretical operational semantics} $\omega_t$ \cite{Duck2004}, akin to a small-step operational semantics. In this case, each rule gives rise to a transition between execution states. Each of these execution states is of the form $\langle G, S, B, T, \mathcal{N} \rangle$ where $G$ is the set of \emph{goal} constraints; $S$ is the constraint \emph{store}, which saves constraints along with an identifier; $B$ is a set of \emph{built-in} constraints; $T$ is the \emph{propagation} history; and $\mathcal{N}$ is the set of nominal constants in use. Built-in constraints are those known internally to the CHR engine and for which the engine may perform reasoning; we assume that an entailment relation $\Vdash$ is given for such constraints.

The rules defining the operational semantics for \chrnabla{} are given in \autoref{fig:theo}. They are quite similar to the original $\omega_t$. 
\begin{itemize}
\item The Solve rule moves a built-in constraint $c$ from the goal set to the built-in set. In practice, this means that the underlying procedure for $c$ is invoked.
\item The Introduce step assigns new identifiers to yet-unsolved goals. Attaching such an identifier is necessary to prevent trivial non-termination arising from using the same rule over the same set of constraints repeatedly.
\item The Apply rule executes a rule with a matching set of constraints. This is where the differences with ``classical'' CHRs arise. First, unification is higher-order. Second, since we have two types of quantification in the body $C$, we must introduce new nominal constants and new variables. We need to keep track of which nominal constants have been introduced in order to deal with confluence (\autoref{sec:confluence}), so we introduce a set $\mathcal{N}$ in the execution state to hold that information. Finally, we might need to $\beta_0$-reduce some of the obtained constraints in order to put them in the right syntax for further transitions.
\end{itemize}
Note that we require that constraints $C'$ obtained after freshening and reduction to be in normal form. That is, they \emph{must} be headed by a constraint constructor; abstraction and application are not allowed at constraint level.

\begin{figure}
\begin{flushleft}
{\bf Solve.} For each built-in constraint $c$,
$$\langle \{ c \} \uplus G, S, B, T, \mathcal{N} \rangle \leadsto \langle G, S, c \wedge B, T, \mathcal{N} \rangle$$
{\bf Introduce.} For each constraint $c$, given a fresh identifier $i$,
$$\langle \{ c \} \uplus G, S, B, T, \mathcal{N} \rangle \leadsto \langle G, \{ c\#i \} \cup S, B, T, \mathcal{N} \rangle$$
{\bf Apply.} If the set of rules contains a rule named $r$ defined as 
$$H^k \; \backslash \; H^r \iff G \;\; | \;\; \nabla \bar{X}. \exists \bar{Y}. C$$
and there is a matching $L_\lambda$-substitution $\theta$ such that $H_1 \equiv \theta(H^k)$, $H_2 \equiv \theta(H^r)$, the set of built-ins $B$ implies $\theta(G)$, that is, $B \Vdash \theta(G)$ and $t = \langle \mathsf{id}(H_1), \mathsf{id}(H_2), r \rangle \not\in T$,
$$
\langle G, H_1 \uplus H_2 \uplus S, B, T, \mathcal{N} \rangle
\leadsto \langle C' \uplus G, H_1 \cup S, \theta \wedge B, T \cup \{ t \}, \mathcal{N} \cup \bar{N} \rangle
$$
where $C'$ results from replacing $\bar{X}$ by fresh nominal constants $\bar{N}$, $\bar{Y}$ by fresh variables, and by $\beta_0$-reducing as much as possible.
\end{flushleft}
\caption{Theoretical operational semantics of \chrnabla}
\label{fig:theo}
\end{figure}

\subsection{Simple Type Inference for Higher-rank Types}
\label{sec:hr}

Higher-rank types extend Hindley-Damas-Milner types by allowing polymorphic types as arguments to function types. The type \ensuremath{\forall\!\;\Varid{a}.\,(\Varid{a}\to \Varid{a})\to \Conid{Int}} is not higher-rank, since quantification is at the top of the type. In contrast, \ensuremath{(\forall\!\;\Varid{a}.\,\Varid{a}\to \Varid{a})\to \Conid{Int}} is a higher-rank type, since \ensuremath{\forall\!\;\Varid{a}.\,\Varid{a}\to \Varid{a}}, a quantified type, is the type of the argument of the function.

Inference in the presence of such types has led to different approaches \cite{DBLP:journals/jfp/JonesVWS07,DBLP:conf/icfp/DunfieldK13}. Our aim in this section is to present a simple algorithm that shows the feasibility of using \chrnabla{} for encoding such type systems. The focus in the presentation is simplicity; although we conjecture that a simple variation of the presented procedure is complete.

First of all, we need to describe the syntax of the terms and constraints we deal with. In this case, terms will correspond to shapes of types:
\begin{itemize}
\itemsep0em
\item Types headed by constructors are represented as \ensuremath{\Varid{con}\;(\Conid{C},\Conid{Args})}, where \ensuremath{\Conid{C}} is the name of the constructor and \ensuremath{\Conid{Args}} a list of arguments. For example, the type \ensuremath{\Conid{Maybe}\;\Conid{Int}} is represented as \ensuremath{\Varid{con}\;(\text{\tt \char34 Maybe\char34},\,\;\!\![\Varid{con}\;(\text{\tt \char34 Int\char34},\,\;\!\![\mskip1.5mu])\mskip1.5mu])}. For the sake of conciseness in the examples, we use special syntax for function types, \ensuremath{\Varid{fn}\;(\Conid{A},\Conid{B})}.
\item Polymorphic types are wrapped into a \ensuremath{\mathit{forall}} term constructor, and have as only argument an abstraction binding a variable. For example, \ensuremath{\forall\!\;\Varid{a}.\,\Varid{a}\to \Varid{a}} is represented as \ensuremath{\mathit{forall}\;(\lambda \Conid{A}.\,\Varid{fn}\;(\Conid{A},\Conid{A}))}.
\item Type variables are represented by CHR-level variables.
\end{itemize}
Constraints take types as arguments: they can either be \emph{instantiation} constraints $T_1 \leq T_2$, meaning that $T_1$ is more polymorphic than $T_2$, or \emph{equality} constraints $T_1 = T_2$. We assume that equality constraints are built in, and perform unification on their arguments.

\begin{figure}[t]
\begin{tabular}{p{0.25\textwidth}p{0.65\textwidth}}
\begin{prooftree}
\AxiomC{$x : \tau \in \Gamma$}
\UnaryInfC{$\Gamma \vdash x : \tau \leadsto \top$}
\end{prooftree}
&
\begin{prooftree}
\AxiomC{$\Gamma \vdash e_1 : \tau_1 \leadsto C_1$}
\AxiomC{$\Gamma \vdash e_2 : \tau_2 \leadsto C_2$}
\AxiomC{$\alpha$ and $\beta$ fresh}
\TrinaryInfC{$\Gamma \vdash e_1 \; e_2 : \beta \leadsto \tau_1 \leq \mathit{fn}(\alpha, \beta), \tau_2 \leq \alpha, C_1, C_2$}
\end{prooftree}
\end{tabular}
\begin{tabular}{p{0.45\textwidth}p{0.45\textwidth}}
\vspace{-0.3cm}
\begin{prooftree}
\AxiomC{$\Gamma, x : \tau_1 \vdash e : \tau_2 \leadsto C$}
\UnaryInfC{$\Gamma \vdash \lambda (x : \tau_1). e : \mathit{fn}(\tau_1, \tau_2) \leadsto C$} 
\end{prooftree}
&
\vspace{-0.3cm}
\begin{prooftree}
\AxiomC{$\Gamma, x : \alpha \vdash e : \tau_2 \leadsto C$}
\AxiomC{$\alpha$ fresh}
\BinaryInfC{$\Gamma \vdash \lambda x. e : \mathit{fn}(\alpha, \tau_2) \leadsto C$} 
\end{prooftree}
\end{tabular}
\caption{Constraint generation for $\lambda$-calculus}
\label{fig:gen}
\end{figure}

Constraint-based type inference \cite{Pottier-Remy/emlti,Hage:2009:SSC:1519534.1519622} is structured in multiple phases. The first phase is constraint \emph{generation}: traversing the abstract syntax tree of the expression we are interested in typing and obtaining the corresponding constraints. We focus on a typed $\lambda$-calculus, whose generation judgment $\Gamma \vdash e : \tau \leadsto C$ is given in \autoref{fig:gen}. This judgement takes as input an environment $\Gamma$ and an expression $e$ and produces a type $\tau$ and some constraints $C$. The only unusual feature of the presented $\lambda$-calculus is the existence of an annotated abstraction. Such a feature is needed to type functions with higher-rank arguments \cite{DBLP:journals/jfp/JonesVWS07}.

\begin{figure}[t]
$$
\begin{array}{rclcrcl}
T & \leq & T & \iff & \mathit{true} \\
\mathit{con}(C_1, \mathit{Args}_1) & \leq & T_2 & \iff & \mathit{con}(C_1, \mathit{Args}_1) & = & T_2 \\
\mathit{fn}(S_1, T_1) & \leq & T_2 & \iff &  \mathit{fn}(S_1, T_1) & = & T_2 \\
\mathit{forall}(Q) & \leq & T_2 & \iff & \exists V. \, Q \; V & \leq & T_2 \quad \text{if } T_2 \not\equiv \mathit{forall}(R) \\
T_1 & \leq & \mathit{forall}(Q) & \iff & \nabla V. T_1 & \leq & Q \; V \\
\end{array}
$$
\caption{Rules for solving type inference constraints}
\label{fig:solv}
\end{figure}

The second phase is constraint \emph{solving}: at this point \chrnabla{} enters the game. We implement the solving procedure as a set of rules, given in \autoref{fig:solv}, to be applied exhaustively to the generated constraints. The first rule implements reflexivity. The reader might wonder whether \ensuremath{\Conid{T}\leq \Conid{T}} is a well-defined constraint: it is so, since the duplicated variable appears as argument to a constructor \ensuremath{\leq }, not as argument to a free variable. The next two rules simplify instantiations in which the left-hand side is not polymorphic to an equality. Note that by choosing these rules we make function types invariant. Finally, polymorphic types are dealt with by generating fresh variables or nominal constants.

After solving, leftover constraints are interpreted. Some of them, like \ensuremath{\Conid{V}\mathrel{=}\Conid{T}} symbolize the type assignment found by the type checker, whereas other types of constraints are expected to be absent in the final set.

Let us check how the type engine proceeds with an expression such as \ensuremath{\Varid{id}\;\mathrm{3}}, for the common \ensuremath{\Varid{id}\mathbin{::}\forall\!\;\Varid{a}.\,\Varid{a}\to \Varid{a}} function. The generation derivation is:
\begin{prooftree}
\AxiomC{$\Gamma \vdash \mathit{id} : \mathit{forall}(\lambda A. \, \mathit{fn}(A,A)) \leadsto \top$}
\AxiomC{$\Gamma \vdash 3 : \mathit{con}(\texttt{"Int"}, []) \leadsto \top$}
\BinaryInfC{$
\Gamma \vdash \mathit{id} \; 3 : T \leadsto \mathit{forall}(\lambda A. \, \mathit{fn}(A,A)) \leq \mathit{fn}(S, T), \mathit{con}(\texttt{"Int"}, []) \leq S
$}
\end{prooftree}
The first constraint is simplified by instantiating a fresh variable \ensuremath{\Conid{R}} to \ensuremath{\Varid{fn}\;(\Conid{R},\Conid{R})\leq \Varid{fn}\;(\Conid{S},\Conid{T})}. Then the instantiation constraint is turned into an equality, efectively unifying all of \ensuremath{\Conid{R}}, \ensuremath{\Conid{S}} and \ensuremath{\Conid{T}}. The second constraint is now \ensuremath{\Varid{con}\;(\text{\tt \char34 Int\char34},\,\;\!\![\mskip1.5mu])\leq \Conid{R}}, which is also simplified to an unification \ensuremath{\Varid{con}\;(\text{\tt \char34 Int\char34},\,\;\!\![\mskip1.5mu])\mathrel{=}\Conid{R}}. At this point, no instantiation constraint is left and the value of every type variable has been assigned by unification.

\subsection{Confluence and Termination Properties}
\label{sec:confluence}

CHRs are non-deterministic, as discussed in \autoref{sec:chr}. In contrast with logic languages such as Prolog, CHRs feature committed choice, that is, once a rule has been appplied there is no built-in backtracking mechanism. Thus, an important question when faced with a set of CHRs is whether they are \emph{confluent}, which means that the same final state is achieved regardless of the order in which rules are applied. The other important question is whether a set of CHRs \emph{terminates} for any given input.

In theory, CHRs could be treated as an instance of an abstract rewriting system. However, the fact that matching involves more than one constraint at once sets them apart from other rewriting systems: both confluence \cite{Abdennadher1997,Duck2007} and termination \cite{Voets2008,Pilozzi2008} require specific techniques. In this section we study the applicability of those techniques to our setting.

\paragraph{Confluence.} In order to define confluence we first need to define when two execution states are thought of as equivalent. We build on the definition given by Duck {\it et al.} \cite{Duck2007}, which we extend to account for nominal constants.
\begin{defn}[Variants]
Let $\sigma_1 = \langle G_1, S_1, B_1, T_1, \mathcal{N}_1 \rangle$ and $\sigma_2 = \langle G_2, S_2, B_2, T_2, \mathcal{N}_2 \rangle$ be two execution states. We assume $\mathcal{N}_1$ and $\mathcal{N}_2$ have the same number of variables, otherwise we can just extend the shorter with new fresh ones. For each execution state $\sigma_i$, let $T'_i$ be the set of tokens from $T_i$ which mention any of the constraints in $S_i$ and let $\mathcal{V}_i$ be the set of variables appearing in $G_i$, $S_i$, $B_i$ or $T_i$. We say that $\sigma_1$ and $\sigma_2$ are \emph{variants}, $\sigma_1 \approx \sigma_2$, if either:
\begin{itemize}
\itemsep0em
\item There exists a unifier $\rho$ of $S_1, S_2, G_1, G_2, T'_1$ and $T'_2$, and a permutation $\pi$ between $\mathcal{N}_2$ and $\mathcal{N}_1$ -- that is, a bijective mapping between the two sets of nominal constants -- such that
$\exists \mathcal{V}_1. \, B_1 \supset \exists \mathcal{V}_1. \, \rho \wedge \pi(B_2)$
 and $\exists \mathcal{V}_2. \, B_2 \supset \exists \mathcal{V}_2. \, \rho \wedge \pi^{-1}(B_1)$.
In other words, there exists a unifier modulo renaming.
\item Or both $B_1$ and $B_2$ are logically inconsistent.
\end{itemize}
\end{defn}
\noindent The idea of being equal modulo permutation comes from the usage of the {\sc id$\pi$} rule introduced in \autoref{sec:nablaintro}. In a logical sense, two execution states which are variants satisfy:
$$
\nabla \mathcal{N}_1. \, \exists \mathcal{V}_1. \, B_1 \supset \nabla \mathcal{N}_1. \, \exists \mathcal{V}_2. \, \rho \wedge B_2
\quad \text{and} \quad
\nabla \mathcal{N}_2. \, \exists \mathcal{V}_2. \, B_2 \supset \nabla \mathcal{N}_2. \, \exists \mathcal{V}_1. \, \rho \wedge B_1
$$
However, the reading using permutations is more operational.

Once the notion of variance is settled, we can formally define the confluence property. In the following, $\leadsto$ refers to the theoretical operational semantics relation defined in \autoref{fig:theo}, and $\leadsto^*$ to its transitive and reflexive closure.

\begin{defn}[Joinable states, local confluence, confluence]
Given two execution states $\sigma_1$ and $\sigma_2$, we say that they are \emph{joinable}, $\sigma_1 \downarrow \sigma_2$, if there exists $\sigma'_1$ and $\sigma'_2$ such that $\sigma_1 \leadsto^* \sigma'_1$, $\sigma_2 \leadsto^* \sigma'_2$ and the end states are variants, $\sigma'_1 \approx \sigma'_2$.

\vspace{0.1cm}
\noindent A set of CHRs is said to be \emph{locally confluent} if for any states $\sigma_0, \sigma_1$ and $\sigma_2$ such that $\sigma_0 \leadsto \sigma_1$ and $\sigma_0 \leadsto \sigma_2$, we have that $\sigma_1$ and $\sigma_2$ are joinable.

\vspace{0.1cm}
\noindent A set of CHRs is said to be \emph{confluent} if for any states $\sigma_0, \sigma_1$ and $\sigma_2$ such that $\sigma_0 \leadsto^* \sigma_1$ and $\sigma_0 \leadsto^* \sigma_2$, we have that $\sigma_1$ and $\sigma_2$ are joinable.
\end{defn}
\noindent Note the difference between the two notions of confluence. In local confluence, we take \emph{one} step and then try to join the new states, whereas in normal confluence we are allowed to take any number of steps in the hypothesis. 
\begin{lemma}[Newman 1942]
If a terminating abstract rewrite system is locally confluent, then it is confluent. 
\end{lemma}
The classical way to prove confluence of CHRs is to determine that your rules are locally confluent and then use Newman's Lemma. Proving local confluence directly is still hard, though. Most works \cite{Abdennadher1997,Duck2007} provide sufficient conditions for a set of CHRs to be locally confluent via the notion of a \emph{critical pair}. Intuitively, a critical pair is a minimal description of a point where confluence is at risk, for example, because more than one rule may apply. As previously, we follow \cite{Duck2007}.
\begin{defn}[Critical pair]
Given two rule instances $H^k_1 \mathrel{\backslash} H^r_1 \iff g_1 \mathrel{|} \nabla \overline{X}_1. \, \exists \overline{Y}_1. \, B_1$ and $H^k_2 \mathrel{\backslash} H^r_2 \iff g_2 \mathrel{|} \nabla \overline{X}_2. \, \exists \overline{Y}_2. \, B_2$, we define the following multisets of constraints, using $\uplus$ for disjoint unions:
$$H^k_1 \uplus H^r_1 = H^\cap_1 \uplus H^\Delta_1 \quad \text{and} \quad H^k_2 \uplus H^r_2 = H^\cap_2 \uplus H^\Delta_2$$
Intuitively, all the constraints involved in a rule, $H^k_i \uplus H^r_i$, are divided into two sets, $H^\cap_i$ and $H^\cap_i$. The constraints in $H^\cap_1$ and $H^\cap_2$ are those where the rules may overlap -- if they do not, the equality in the definition of the critical pair becomes false. $H^\Delta_1$ and $H^\Delta_2$, on the other hand, represent those constraints where rules do not coincide. By ranging over all possible partitions of $H^\cap_i$ and $H^\Delta_i$, we consider all possible scenarios in which both rules could apply.

We define the propagation history $T_{CP}$ to include a token $e_i = \langle \mathsf{id}(H^k_i), \mathsf{id}(H^r_i), r_i \rangle$ for each propagation rule. A \emph{critical pair} for these two rules is the pair of execution states:
$$\begin{array}{l}
\langle B_1, (H^\Delta_1 \uplus H^\Delta_2 \uplus H^\cap_1) - H^r_1, H^\cap_1 = H^\cap_2 \wedge g_1 \wedge g_2, T_{CP} - \{ e_1 \}, \mathcal{N}_1 \cup \mathcal{N}_2 \cup \overline{X}_1 \rangle \\
\langle B_2, (H^\Delta_1 \uplus H^\Delta_2 \uplus H^\cap_1) - H^r_2, H^\cap_1 = H^\cap_2 \wedge g_1 \wedge g_2, T_{CP} - \{ e_2 \}, \mathcal{N}_1 \cup \mathcal{N}_2 \cup \overline{X}_2 \rangle
\end{array}$$
where $\mathcal{N}_i$ is the set of nominal constants in each rule instance. 
\end{defn}
\noindent A critical pair encodes the result of applying the two rules starting from an initial execution state to which both rules may be applied.
\begin{restatable}[Joinable critical pairs $\implies$ locally confluent]{theorem}{thmmain}
Given a set $\mathcal{C}$ of CHRs, if all critical pairs are joinable, then $\mathcal{C}$ is locally confluent.
\end{restatable}
\begin{proof}
See Appendix.
\end{proof}
\begin{corollary}[Joinable critical pairs + terminating $\implies$ confluent]
Given a terminating set $\mathcal{C}$ of CHRs, if all critical pairs are joinable, then $\mathcal{C}$ is confluent.
\end{corollary}

As an example of how to apply this technique, let us look at a critical pair for the CHRs in \autoref{fig:solv}. In particular, when instantiation involves a type of the form \ensuremath{\mathit{forall}\;(\lambda X_1.\,\mathit{forall}\;(\lambda X_2.\,\mathbin{...}\mathit{forall}\;(\lambda X_n.\,\Conid{Q})))} both the first rule, which encodes reflexivity, and the last rule, which tells us what to do when the right-hand side is a quantified type, are applicable. In the first case the resulting state is $\top$, and in the second case,
\begin{tabbing}
\qquad\=\hspace{\lwidth}\=\hspace{\cwidth}\=\+\kill
${\mathit{forall}\;(\lambda X_1.\,\mathit{forall}\;(\lambda X_2.\,\mathbin{...}\mathit{forall}\;(\lambda X_n.\,\Conid{Q})))\leq \mathit{forall}\;(\lambda X_2.\,\mathbin{...}\mathit{forall}\;(\lambda X_n.\,\Conid{Q}\;\mathtt{a}_1))}$
\end{tabbing}We need to check that these two states are joinable. The first resulting state does not allow more rules to be applied, but for the second we can keep going until we end up with a constraint of the form
\ensuremath{\Conid{Q}\;\alpha_1\;\alpha_2\mathbin{...}\alpha_n\mathrel{=}\Conid{Q}\;\mathtt{a}_1\;\mathtt{a}_2\mathbin{...}\mathtt{a}_n}
for fresh existential variables $\alpha_i$ and fresh constants $\mathsf{a}_i$. Reading the definition of variants, we need to prove that
$$
\top \supset \exists \overline{\alpha}_i. \, Q \; \alpha_1 \; \dots \; \alpha_n = Q \; \mathsf{a}_1 \; \dots \; \mathsf{a}_n
\quad \text{and} \quad
\exists \overline{\alpha}_i. \, Q \; \alpha_1 \; \dots \; \alpha_n = Q \; \mathsf{a}_1 \; \dots \; \mathsf{a}_n \supset \top
$$
The second one is trivial. For the first one, just take $\alpha_i$ equal to $\mathsf{a}_i$ and we are done. Since we have not introduced any constant, we can take any permutation of the $\mathsf{a}_i$'s to make the execution states variants.

\paragraph{Termination.} Most approaches to termination of CHRs are based on a quantity which decreases after each step of solving. A norm is defined by \cite{Fruhwirth2000,Voets2008,Pilozzi2008} as a function $\| \cdot \| : \ensuremath{\Conid{Term}} \to \mathbb{N}$ and a level mapping as a function $| \cdot | : \ensuremath{\Conid{Constraint}} \to \mathbb{N}$. If some conditions over the rules and initial constraints referring to those mappings are satisfied, then the CHRs are terminating for a given initial set of constraints.

A key property in most realizations of these conditions is that of a \emph{rigid} level mapping. A rigid $| \cdot |$ is invariant under substitution: for all constraints $C$ and substitutions $\theta$, $| C | = | \theta C |$. We can redefine this notion in such a way that the techniques of \cite{Fruhwirth2000,Voets2008,Pilozzi2008} keep working:
\begin{itemize}
\itemsep0em
\item Level mappings assume that constraints are $\beta_0$-reduced before application.
\item Substitutions $\theta$ map variables to possibly higher-order terms.
\item The mapping must be independent of nominal constraints. Formally, for every permutation of variables $\pi$, $| C | = | \pi(C) |$. Otherwise a rule like $c(X) \iff \nabla Y. \, c(Y)$ may generate an infinite number of new constants.
\end{itemize}

Using these ingredients, \cite{Voets2008} shows that the following ranking condition is enough to guarantee termination. \cite{Fruhwirth2000} and \cite{Pilozzi2008} describe different conditions in the same spirit.
\begin{defn}[Call set]
Given a set of CHRs $R$ and an initial set of constraints $I$, the \emph{call set} $\mathsf{Call}(R,I)$ is the union, over all possible traces of computations of $R$ on $I$, of all constraints added to the store at every rule application.
\end{defn}
\begin{defn}[Ranking condition for CHRs]
A set of CHRs $R$ and an initial set of constraints $I$ is said to satisfy the \emph{ranking condition} with respect to a level mapping $| \cdot |$ if every constraint in $\mathsf{Call}(R,I)$ is rigid with respect to $| \cdot |$, and for each rule $R$, matching substitution $\theta$ and answer substitution $\varphi$:
\begin{itemize}
\item If $R$ is a propagation rule $H_1, \dots, H_n \implies G \;\; | \;\; B_1, \dots, B_m$, for every $i = 1, \dots, n$ and $j = 1, \dots, m$ such that $B_j$ is not built-in, then $|\theta H_i | > |\varphi \theta B_j|$
\item If $R$ is a simpagation rule $H \;\; \backslash \;\; H_1, \dots, H_n \implies G \;\; | \;\; B_1, \dots, B_m$, define
$$p = \max \{ |\theta H_1|, \dots, |\theta H_n|, |\varphi \theta B_1|, \dots, |\varphi \theta B_m|  \;\; | \;\; B_j \text{ not built-in} \}$$
Then the number of constraints with level value $p$ in $\{ \theta H_1, \dots, \theta H_n \}$ is higher than the the number of those constraints in $\{ \varphi \theta B_1, \dots, \varphi \theta B_m \}$.
\end{itemize}
\end{defn}

As an example, in the rules for type inference given in \autoref{fig:solv}, the only problematic rule with respect to nominal constants is the last one, since the others do not introduce any constants. We can use a norm which puts equality constraints before instantiation constraints, and for the latter type, adds up the number of \ensuremath{\mathit{forall}} constituents from both sides. The reader can check that this norm always decreases during solving.

\section{Implementation}
\label{sec:impl}

We have built a solver for \chrnabla{} as a deeply embedded language in Haskell; it is available at \url{https://git.science.uu.nl/f100183/uchrp}. Apart from the features described in this paper, it also provides a type-safe interface to CHRs and support for rule priorities as described by \cite{DeKoninck:2007:URP:1273920.1273924}. The only difference between the library and this paper is that the former uses a built-in predicate to introduce variables, instead of making the operator $\nabla$ part of the syntax of rules.

Using this library, we have implemented the type checker featured in this paper to show the feasibility of the approach; it is available at \url{https://git.science.uu.nl/f100183/quique}. In fact, this type checker also implements impredicative polymorphism, although the description of the corresponding set of CHRs is out of scope for this paper. This suggests that \chrnabla{} is expressive enough for encoding complex type systems.

Our preliminary evaluation shows that \chrnabla{} is competitive in terms of performance with similar approaches. Furthermore, there is room for optimizations such as performing on-the-fly substitution and indexing constraints in the solver.

\section{Related Work}

\paragraph{$\lambda$-tree syntax and $\nabla$ quantification.} The closest system to ours is $\lambda$Prolog \cite{Miller91alogic,Miller:2012:PHL:2331097}. We have drawn inspiration from several of its features, like the use of $\lambda$-tree syntax and pattern unification in heads. The main difference is that $\lambda$Prolog, as its name suggests, builds upon the logic programming paradigm, so it includes search with backtracking. On the other hand, \chrnabla{} embodies committed choice: no backtracking is done once a rule is chosen to apply. As a result, the kind of formulas we can represent in \chrnabla{} is more limited than in $\lambda$Prolog: the latter allows hereditary Harrop formulas, where both universal quantification and implication may appear in the left-hand side of an implication, whereas we do not allow left-nested implications.

The Abella theorem prover \cite{JFR4650} was developed to reason about $\lambda$Prolog and inherits many of its features. Abella uses the $\nabla$ quantifier to reason about the object logic into its meta-logic. Our approach is similar, but we have blurred the lines between object and meta-logic: in \chrnabla{} only universal quantification in using nominal constants is available. This results in a more uniform approach to deal with constraints. 

In the functional world, ML has been extended to support binders using an approach similar to $\lambda$-tree syntax \cite{tr-ml-binder}. A type \ensuremath{\Varid{a}\Rightarrow \Varid{b}} represents values of type \ensuremath{\Varid{b}} where a variable of type \ensuremath{\Varid{a}} is bound. Matching on a value of this type binds to a function.

\paragraph{Nominal abstract syntax.} Another different approach to introduce binders in the syntax of a language is given by \emph{nominal syntax}. The main difference is that correctness is based on the notion of \emph{variable swapping}. Using this idea, it can be proven that programs operate correctly under $\alpha$-equivalence. See \cite{TIU20073} for a detailed account of the differences.

 $\alpha$Prolog extends Horn clauses with nominal binding \cite{Cheney:2008:NLP:1387673.1387675}. FreshML \cite{Shinwell:2003:FPB:944705.944729} implements these ideas in the functional world.

\section{Conclusion and Future Work}

\chrnabla{} provides an extension to Constraint Handling Rules in which binding manipulation is reflected in the syntax. This provides a sound basis for developing compilers and type checkers for advanced type programming languages.

In the future, we aim to look at other extensions of CHRs and their integration with binding and $\nabla$ quantification. In particular CHR${}^\vee$ \cite{Abdennadher:1998:CFQ:645423.756455}, which features disjunction, enables us to escape from the committed choice semantics of CHR and explore several branches of computation. Such an ability is needed for type checking some forms of overloading, like the one present in Swift \cite{swift}.


\bibliographystyle{acmtrans}
\bibliography{full}

\newpage
\appendix

\section{Proof of the Confluence Theorem}
\label{sec:proof}

\thmmain*
\begin{proof}
The proof is similar to the one in \cite{Abdennadher1997}, we need to consider each possible pair of execution steps in the theoretical operational semantics $\omega_t$.

\paragraph{Solve + Solve.} In this case we start with an execution state $\langle \{ c_1,c_2 \} \uplus G, S, B, T, \mathcal{N} \rangle$ for which two different built-ins are moved into the third component:
$$\langle \{ c_2 \} \uplus G, S, \{ c_1 \} \wedge B, T, \mathcal{N} \rangle \; \text{ and } \; \langle \{ c_1 \} \uplus G, S, c_2 \wedge B, T, \mathcal{N} \rangle$$
Now we can choose to move the other built-in to restore variance (actually, equality) of the execution states.

\paragraph{Introduce + Introduce, Solve + Introduce.} Similar to the Solve + Solve case, since the variance of the execution states is independent of the other in which we perform those steps. As in the previous case, we have equality of the joined states.

\paragraph{Solve + Apply, Introduce + Apply.} By inspection of the rules for the semantics, we see that the constraints affected by Solve and Introduce come from the $G$ component of the execution state, whereas those affected by Apply come from $S$. This means that it is not possible for the execution steps to affect the same constraint. As a result, if we take a step in one direction we can always apply the other one irrespectively of the former, leading to variant states.

\paragraph{Apply + Apply.} This is the interesting case. Suppose that the following rule instances have been applied:
$$H^k_1 \mathrel{\backslash} H^r_1 \iff g_1 \mathrel{|} \nabla \overline{X}_1. \, \exists \overline{Y}_1. \, B_1 \; \text{ and } \; H^k_2 \mathrel{\backslash} H^r_2 \iff g_2 \mathrel{|} \nabla \overline{X}_2. \, \exists \overline{Y}_2. \, B_2$$
Following \cite{Abdennadher1997} we further distinguish between two scenarios.

\begin{itemize}
\item No constraint in the first rule unified with another in the second rule (``disjoint peaks''). In this case the application of one rule does not interfere with the application of the other. Thus, given the states $S_1$ and $S_2$ resulting from applying the first and second rule, respectively, we can join them by applying the second rule to $S_1$ and the first rule to $S_2$.

The only problem is finding the permutation between the nominal constants introduced by each rule. Take $\mathcal{N}_1$ to be the set of constants introduced by the first rule in the first state, $\mathcal{N}_2$ those introduced then by applying the second rule, and the same for $\mathcal{M}_1$ and $\mathcal{M}_2$ with respect to $S_2$. We just need to permute $\mathcal{N}_1$ to $\mathcal{M}_1$ and $\mathcal{N}_2$ to $\mathcal{M}_2$ to satisfy the variance requirements.

\item At least one constraint in the first rule unifies with another in the second rule (``critical peaks''). In this case the condition on joinability of critical pairs is enough to guarantee joinability of the corresponding execution states: all the other constraints not mentioned in the critical pair remain untouched by the execution of the rules.
\end{itemize}
\end{proof}

\hide{
\newpage

\section{Main Changes for Second Round Review}

\begin{itemize}
\item In the title we changed ``$\nabla$ Quantification'' to ``Generic Quantification'' -- this is the name used in \cite{TIU20073} -- to address the comment of Reviewer \#2.
\item \autoref{sec:nabla} has been expanded to include a description of all the rules in the theoretical operational semantics, instead of just referring to \cite{Duck2004} and naming the changes, adressing the comments of Reviewer \#2.
\item We make explicit what patterns stand for in \autoref{sec:nabla} -- just another name for constraints.
\item The proof of confluence in Appendix A is now complete, instead of just referring to \cite{Abdennadher1997}, as requested by Reviewer \#1.
\item A more detailed account of termination is given in \autoref{sec:confluence}.
\item Following Reviewer \#3's advice, \autoref{sec:impl} includes now more details about our implementation of higher-rank polymorphism using \chrnabla. In particular, we mention that the performance is acceptable and could be improved with the usual means -- on-the-fly substitution, indexing constraints -- and we conjecture that the algorithm is complete for predicative instantiation.
\end{itemize}
}

\label{lastpage}

\end{document}